\documentclass{article} 

\usepackage{graphicx}

\usepackage[english]{babel}
\usepackage{amsmath}
\usepackage{amssymb}
\usepackage{amsfonts}
\usepackage{amsthm}
\usepackage{mathrsfs}
\usepackage{enumitem}
\usepackage{authblk}

\usepackage{wasysym}

\usepackage{multicol}
\usepackage{booktabs}

\usepackage[round]{natbib}
\usepackage{verbatim}

\renewcommand{\geq}{\geqslant}
\renewcommand{\phi}{\varphi}
\newcommand{\Lang}{\mathcal L}

\newcommand{\D}{\mathcal D}
\newcommand{\B}{\text{\boldmath $B$}}
\renewcommand{\P}{\text{\boldmath $P$}}

\newcommand{\X}{\mathcal X}
\newcommand{\I}{\mathcal I}
\newcommand{\PS}{\text{\it PS}}
\newcommand{\ic}{\text{\sc IC}}
\newcommand{\N}{\mathcal N}

\newcommand{\mo}{\text{\rm Mod}}

\newtheorem{theorem}{Theorem}

\theoremstyle{definition}
\newtheorem{definition}{Definition}

\begin{document}

\title{The Common Structure of Paradoxes in\\ Aggregation Theory}
\author{Umberto Grandi}
\affil{    Department of Mathematics\\
	   University of Padova\\
		Italy \\
	     umberto.uni@gmail.com}

\maketitle

\begin{abstract}
In this paper we analyse some of the classical paradoxes in Social Choice Theory (namely, the Condorcet paradox, the discursive dilemma, the Ostrogorski paradox and the multiple election paradox) using a general framework for the
study of aggregation problems called binary aggregation with integrity constraints. 
We provide a definition of paradox that is general enough to account for the four cases mentioned, and identify a common structure in the syntactic properties of the rationality assumptions that lie behind such paradoxes.
We generalise this observation by providing a full characterisation of the set of rationality assumptions on which the majority rule \emph{does not} generate a paradox.

\end{abstract}



\section{Introduction}


Most work in Social Choice Theory started with the observation of paradoxical situations. 
From the Marquis de Condorcet (\citeyear{Condorcet}) to more recent American court cases \citep{KornhauserSager1986}, a wide collection of paradoxes have been analysed and studied in the literature on Social Choice Theory \citep[see, e.g.,][]{Nurmi}.
More recently, researchers in Artificial Intelligence and in particular from the novel research area of Computational Social Choice \citep[see, e.g.,][]{BrandtEtAlMAS2013} have become interested in the study of collective choice problems in which the set of alternatives has a combinatorial structure. 
Novel paradoxical situations emerged from the study of these situations, and the combinatorial structure of the domains gave rise to interesting computational challenges \citep[][]{ChevaleyreEtAlAIMag2008}.

This paper concentrates on the use of the majority rule in binary combinatorial domains, and investigates the question of what constitutes a paradox in such a setting.
We identify a common structure behind the most classical paradoxes in Social Choice Theory, putting forward a general definition of paradox in aggregation theory. 
By characterising paradoxical situations by means of computationally recognisable properties, we aim at providing more domain-specific research with new tools for the development of safe procedures for collective decision making.

We base the analysis on our previous work on binary aggregation with integrity constraints \citep{GrandiEndrissIJCAI2011, GrandiEndrissAIJ2013}, which constitutes a general framework for the study of aggregation problems. 
In this setting, a set of individuals needs to take a decision over a set of binary issues, and these choices are then aggregated into a collective one.
Given a rationality assumption that binds the choices of the individuals, we define a paradox as a situation in which all individuals are rational but the collective outcome is not.
We present some of the most well-known paradoxes that arise from the use of the majority rule in different contexts, and we show how they can be expressed in binary aggregation as instances of this general definition.
Our analysis focuses on the Condorcet paradox~(\citeyear{Condorcet}), the discursive dilemma in judgment aggregation \citep{KornhauserSager1986, ListPettit2002}, the Ostrogorski paradox (\citeyear{Ostrogorski}) and the more recent work of \citet{BramsEtAl1998} on multiple election paradoxes.

Such a uniform representation of the most important paradoxes in Social Choice Theory enables us to make a crucial observation concerning the syntactic structure of the rationality assumptions that lie behind these paradoxes. 
We represent rationality assumptions by means of logical formulas in a simple propositional language, and we observe that all formulas formalising a number of classical paradoxes feature a disjunction of literals of size at least~$3$. 
This observation can be generalised to a full characterisation of the rationality assumptions on which the majority rule does not generate a paradox, and in Theorem~\ref{thm:icmajority} we identify them as those formulas that are equivalent to a conjunction of clauses of size at most~$2$.

This work needs to be positioned in the growing literature on judgment aggregation \citep[see, e.g.,][]{ListPuppe2009}, and we discuss in a dedicated section the relation between classical frameworks for judgment aggregation and our setting.
While many of the findings shown in this paper can be traced back to known results from this literature, to the best of our knowledge this paper presents the first comprehensive study of paradoxes in aggregation theory and puts forward a simple yet general framework for a unified analysis of such recurrent situations.

The paper is organised as follows. 
In Section~\ref{sec:definition} we give the basic definitions of the framework of binary aggregation with integrity constraints, and we provide a general definition of paradox. 
In Section~\ref{sec:unify} we show how a number of paradoxical situations in Social Choice Theory can be seen as instances of our general definition of paradox, and we identify a syntactic property that is common to all paradoxical rationality assumptions.
Section~\ref{sec:majority} provides a syntactic characterisation of the paradoxical situations for the majority rule, and Section~\ref{sec:conclusions} concludes the paper.


\section{Binary Aggregation with Integrity Constraints}\label{sec:definition}


In this section we provide the basic definitions of the framework of binary aggregation with integrity constraints which we developed in previous work \citep{GrandiEndrissIJCAI2011} based on work by \citet{Wilson1975} and Dokow and Holzman \citep{DokowHolzman2009, DokowHolzman2008}.
In this setting, a number of individuals each need to make a yes/no choice regarding a number of issues and these choices then need to be aggregated into a collective choice. 
Paradoxical situations may occur when a set of individual choices that is considered rational leads to a collective outcome which fails to satisfy the same rationality assumption of the individuals.


\subsection{Terminology and Notation}

Let $\mathcal{I}=\{1,\dotsc,m\}$ be a finite set of \emph{issues}, and let $\mathcal{D}=D_1\times\dots\times D_m$ be a \emph{boolean combinatorial domain}, i.e., $|D_j|=2$ for all $j\in\mathcal{I}$. Without loss of generality we assume that $D_j=\{0,1\}$ for all $j$. Thus, given a set of issues $\I$, the domain associated with it is $\D=\{0,1\}^\I$. A \emph{ballot}\index{ballot} $B$ is an element of $\D$.

In many applications it is necessary to specify which elements of the domain are rational and which should not be taken into consideration. Propositional logic provides a suitable formal language to express possible restrictions of rationality on binary combinatorial domains. 
If $\I$ is a set of $m$ issues, let $\PS=\{p_1,\dotsc,p_m\}$ be a set of propositional symbols, one for each issue, and let $\Lang_{\PS}$ be the propositional language constructed by closing $\PS$ under propositional connectives. 
For any formula $\varphi\in\Lang_{\PS}$, let $\mo(\varphi)$ be the set of assignments that satisfy~$\varphi$. For example, $\mo(p_1\wedge\neg p_2)=\{(1,0,0),(1,0,1)\}$ when $\PS=\{p_1,p_2,p_3\}$. 
An \emph{integrity constraint} is any formula $\ic\in\Lang_{\PS}$. 

Integrity constraints can be used to define what tuples in $\D$ we consider \emph{rational} choices. Any ballot $B\in\D$ is an assignment to the variables $p_1,\dotsc,p_m$, and we call $B$ a \emph{rational ballot} if it satisfies the integrity constraint $\ic$, i.e., if $B$ is an element of $\mo(\ic)$. 
In the sequel we shall use the terms ``integrity constraints'' and ``rationality assumptions''\index{rationality assumption} interchangeably.

Let $\N=\{1,\dotsc,n\}$ be a finite set of \emph{individuals}. We make the assumption that there are at least 2~individuals.
Each individual submits a ballot $B_i \in \D$ to form a \emph{profile} $\B=(B_1,\dotsc,B_n)$. 
We write $b_j$ for the $j$th element of a ballot $B$, and $b_{i,j}$ for the $j$th element of ballot $B_i$ within a profile $\B=(B_1,\dotsc,B_n)$.
Given a finite set of issues $\I$ and a finite set of individuals~$\N$, an \emph{aggregation procedure} is a function $F:\D^\N\to\D$, mapping each profile of binary ballots to an element of~$\D$. Let $F(\B)_j$ denote the result of the aggregation of profile $\B$ on issue~$j$.


\subsection{A General Definition of Paradox}

Consider the following example: Let $\ic=p_1\wedge p_2\rightarrow p_3$ and suppose there are three individuals, choosing ballots $(0,1,0)$, $(1,0,0)$ and $(1,1,1)$. 
Their choices are rational (they all satisfy $\ic$). 
However, if we employ the \emph{majority rule}, i.e., we accept an issue $j$ if and only if a majority of individuals do, we obtain the ballot $(1,1,0)$ as collective outcome, which fails to be rational. This kind of observation is often referred to as a paradox.

We now give a general definition of paradoxical behaviour of an aggregation procedure in terms of the violation of certain rationality assumptions.

\begin{definition}\label{def:paradox}
A \emph{paradox} is a triple $(F,\B,\ic)$, where $F$ is an aggregation procedure, $\B$ is a profile in $\D^\N$, $\ic$ is an integrity constraint in $\Lang_{\PS}$, and $B_i\in\mo(\ic)$ for all~$i\in\N$ but $F(\B)\not\in\mo(\ic)$.
\end{definition}

\noindent
A closely related notion is that of collective rationality:

\begin{definition}\label{def:CR}
Given an integrity constraint $\ic\in\Lang_{\PS}$, an aggregation procedure $F$ is called \emph{collectively rational} (CR) with respect to $\ic$, if for all rational profiles $\B\in\mo(\ic)^\N$ we have that $F(\B)\in \mo(\ic)$.
\end{definition}

\noindent
Thus, $F$ is CR with respect to $\ic$ if it \emph{lifts} the rationality assumption given by $\ic$ from the individual to the collective level, i.e., if $F(\B)\in\mo(\ic)$ whenever $B_i\in\mo(\ic)$ for all $i\in\N$. An aggregation procedure that is CR with respect to $\ic$ cannot generate a paradoxical situation with $\ic$ as integrity constraint.


\subsection{Related Work in Binary Aggregation}


\citet{Wilson1975} has been the first to define and study the framework of binary aggregation. His seminal paper contains several impossibility results for aggregation procedures satisfying an axiomatic property known as \emph{independence}, including a generalisation of the famous impossibility result by \citet{Arrow1963}. 
Wilson's notion of responsive aggregator corresponds to our notion of collective rationality with respect to a family of integrity constraints. \citet{RubinsteinFishburn86} generalised Wilson's framework allowing individuals to choose elements of certain vector spaces. The case of binary aggregation is subsumed by considering the vector space $\D=\{0,1\}^\I$.

A similar setting has been investigated more recently by \citet{DokowHolzman2009, DokowHolzman2008}. 
Their definition of collective rationality is the same as Wilson's, although they consider a single subdomain $\mathcal X\subseteq \D$ of rational ballots at a time rather than a family of such subsets. 
Note that propositional logic is fully expressive with respect to subsets of $\D$, i.e., for every subset $X\subseteq \D$ there exists a formula $\phi_X$ such that $\mo(\phi_X)=X$, hence our approach is equivalent to that of Dokow and Holzman. 
Our choice of using formulas rather than sets is motivated by the possibility of classifying integrity constraints by means of syntactic properties and by the the compactness of this representation. 

Another framework for binary aggregation has been proposed by Nehring and Puppe \citep{NehringPuppe2007, NehringPuppe2010}. Although their aim is more general, they also concentrate on the study of aggregation procedures over property spaces, a setting that is closer to the original framework of \citet{Wilson1975}. 

An important, although not substantial, difference between our framework and classical approaches to binary aggregation resides in our definition of aggregation procedure. Both \citet{DokowHolzman2008} and \citet{NehringPuppe2007} define an aggregation procedure on a specific domain $\mathcal X\subseteq\{0,1\}^m$, including in this definition the notion of collective rationality with respect to the integrity constraint that defines $\mathcal X$. 
The same approach is also used in the literature on judgment aggregation \citep{ListPuppe2009}. 
Instead, we define aggregation procedures on all possible profiles, studying collective rationality as an additional property of an aggregator. 
Our choice is motivated by an attempt to separate the definition of an aggregation procedure and its axiomatic properties from the notion of collective rationality, which depends on the domain of rational ballots on which the aggregation is performed. 

In several papers \citep[see, e.g.,][]{ListPuppe2009, NehringPuppe2010, DokowHolzman2008} it has been observed that the framework of judgment aggregation for propositional logic is equivalent to that of binary aggregation. In Section~\ref{sec:jaba} we discuss in detail the relation between judgment aggregation and binary aggregation.



\section{Unifying Paradoxes in Binary Aggregation}\label{sec:unify}


In this section we present a number of classical paradoxes from Social Choice Theory, and we show how they can be seen as instances of our Definition~\ref{def:paradox}.
In Section~\ref{sec:paba} we introduce the Condorcet paradox, and we show how settings of preference aggregation can be seen as instances of binary aggregation by devising a suitable integrity constraint. 
Section~\ref{sec:jaba} repeats this construction for the framework of judgment aggregation and for the discursive dilemma. 
In Section~\ref{sec:ostrogorski} we then deal with the Ostrogorski paradox, in which a paradoxical feature of representative majoritarian systems is analysed, and in Section~\ref{sec:commonstructure} we conclude by identifying a common structure in the integrity constraints that lie behind those paradoxes. 
Section~\ref{sec:mep} presents two further paradoxes that occur when voting with multiple issues.


\subsection{The Condorcet Paradox\index{paradox!Condorcet paradox} and Preference Aggregation}\label{sec:paba}


During the Enlightment period in France, several active scholars dedicated themselves to the problem of collective choice, and in particular to the creation of new procedures for the election of candidates. Although these are not the first documented studies of the problem of social choice \citep{LeanUrken}, Marie Jean Antoine Nicolas de Caritat, the Marquis de Condorcet, was the first to point out a crucial problem of the most basic voting rule that was being used, the majority rule (Condorcet, \citeyear{Condorcet}). The paradox he discovered, that now comes under his name, is explained in the following paragraphs:

\begin{quote}
\textbf{Condorcet Paradox}. Three individuals need to decide on the ranking of three alternatives $\{\triangle,\ocircle,\square\}$. Each individual expresses her own ranking and the collective outcome is aggregated by pairwise majority: an alternative is preferred to a second one if and only if a majority of the individuals prefer the first alternative to the second. Consider the following situation:  

{\large
\[\begin{array}{c}
\triangle <_1 \ocircle <_1 \square \\
\square <_2 \triangle <_2 \ocircle \\
\ocircle <_3 \square <_3 \triangle \\
\rule{3.1cm}{.2mm} \\
\triangle < \ocircle < \square < \triangle \\
\end{array}\]
}

When computing the outcome of the pairwise majority rule, we notice that there is a majority of individuals preferring the circle to the triangle ($\triangle<\ocircle$); that there is a majority of individuals preferring the square to the circle ($\ocircle<\square$); and, finally, that there is a majority of individuals preferring the triangle to the square ($\square<\triangle$). The resulting outcome fails to be a linear order, giving rise to a circular collective preference between the alternatives.
\end{quote}

\subsubsection{Preference Aggregation}

Condorcet's paradox was rediscovered in the second half of the XXth century while a whole theory of \emph{preference aggregation} was being developed \citep[see, e.g.,][]{Gaertner}.
This framework considers a finite set of individuals $\N$ expressing preferences over a finite set of alternatives $\X$. A preference relation is represented by a binary relation over $\X$. 
Preference relations are traditionally assumed to be \emph{weak orders}, i.e., reflexive, transitive and complete binary relations. 
Another common assumption is representing preferences as \emph{linear orders}, i.e., irreflexive, transitive and complete binary relations. 
In the sequel we shall assume that preferences are represented as linear orders, writing $aPb$ for ``alternative $a$ is strictly preferred to $b$''. 
Each individual in $\N$ submits a linear order $P_i$, forming a profile $\P=(P_1,\dots,P_{|\N|})$.
Let $\Lang(\X)$ denote the set of all linear orders on $\X$. 
Given a finite set of individuals $\N$ and a finite set of alternatives $\X$, a \emph{social welfare function} is a function $F:\Lang(\X)^\N\to\Lang(X)$.



\subsubsection{Translation}\label{sec:PAtranslation}

Given a preference aggregation problem defined by a set of individuals $\N$ and a set of alternatives $\X$, let us consider the following setting for binary aggregation. Define a set of issues $\I_{\X}$ as the set of all pairs $(a,b)$ in $\X$. The domain $\D_{\X}$ of aggregation is $\{0,1\}^{|\X|^2}$. In this setting, a binary ballot $B$ corresponds to a binary relation $P$ over $\X$: $B_{(a,b)}=1$ if and only if $a$ is in relation to $b$ ($aPb$). Given this representation, we can associate with every SWF for $\X$ and $\N$ an aggregation procedure that is defined on a subdomain of $\D_{\X}^\N$. We now characterise this domain as the set of models of a suitable integrity constraint.

Using the propositional language $\Lang_{\PS}$ constructed over the set $\I_\X$, we can express properties of binary ballots in $\D_{\X}$. In this case $\Lang_{\PS}$ consists of $|\X|^2$ propositional symbols, which we call $p_{ab}$ for every issue $(a,b)$.
The properties of linear orders can be enforced on binary ballots using the following set of integrity constraints, which we shall call $\ic_<$:\footnote{We will use the notation $\ic$ both for a single integrity constraint and for a set of formulas---in the latter case considering as the actual constraint the conjunction of all the formulas in $\ic$.}

\begin{description}[noitemsep]
\item \textbf{Irreflexivity}: $\neg p_{aa}$ for all $a\in\X$ 
\item \textbf{Completeness}:  $p_{ab}\vee p_{ba}$ for all $a\neq b\in \X$
\item \textbf{Transitivity}:  $p_{ab}\wedge p_{bc}{\rightarrow} p_{ac}$ for  $a,b,c\in \X$ pairwise distinct
\end{description}

\noindent
In case preferences are expressed using weak orders rather than linear orders, it is sufficient to replace the integrity constraints of irreflexivity in $\ic_<$ with their negation to obtain a similar correspondence between SWFs and aggregation procedures.  By dropping the axiom of completeness instead, we obtain preference aggregation with partial orders. Many other classical properties of preferences can be expressed with this formalism. A notable example is the property of \emph{negative transitivity}, which is expressed by the following integrity constraint: $\neg p_{ab} \wedge \neg p_{bc} \rightarrow \neg p_{ac}$ for $a,b,c\in X$ pairwise distinct.



\subsubsection{The Condorcet Paradox in Binary Aggregation}

\begin{table}
\begin{center}
\label{table:condo}
\begin{tabular}{cccc}
\hline\noalign{\smallskip}
 & $\triangle\ocircle$ & $\ocircle\square$ & $\triangle\square$ \\
\noalign{\smallskip}\hline\noalign{\smallskip}
Voter 1 & 1 & 1 & 1 \\
Voter 2 & 1 & 0 & 0 \\
Voter 3 & 0 & 1 & 0 \\
\noalign{\smallskip}\hline
\emph{Maj} & 1 & 1 & 0 
 \end{tabular}
\end{center}
\caption{The Condorcet paradox in binary aggregation.}
\end{table}

The translation presented in the previous section enables us to express the Condorcet paradox in terms of Definition~\ref{def:paradox}. Let $\X=\{\triangle,\ocircle,\square\}$ and let $\N$ contain three individuals. Consider the profile $\B$ for $\I_\X$ in the Table~\ref{table:condo}, where we have omitted the values of the reflexive issues $(\triangle,\triangle)$ (always 0 by $\ic_<$), and specified the value of only one of $(\triangle,\ocircle)$ and $(\ocircle,\triangle)$ (the other can be obtained by taking the opposite of the value of the first), and accordingly for the other alternatives.
Every individual ballot satisfies $\ic_<$, but the outcome obtained using the majority rule \emph{Maj} (which corresponds to pairwise majority in preference aggregation) does not satisfy $\ic_<$: the formula $p_{\triangle\ocircle}\wedge p_{\ocircle\square}\rightarrow p_{\triangle\square}$ is falsified by the outcome. Therefore, $(\text{\emph{Maj}},\B,\ic_<)$ is a paradox by Definition~\ref{def:paradox}.


\subsection{The Discursive Dilemma and Judgment Aggregation}\label{sec:jaba}


The discursive dilemma emerged from the formal study of court cases that was carried out in recent years in the literature on law and economics, generalising the observation of a paradoxical situation known as the ``doctrinal paradox'' \citep{KornhauserSager1986}. 
Such a setting was first given mathematical treatment by \citet{ListPettit2002}, giving rise to an entirely new research area in Social Choice Theory known as \emph{judgment aggregation}. 
Earlier versions of this paradox can be found in work by \citet{Guilbaud1952} and \citet{Vacca1922}. We now describe one of the most common versions of the discursive dilemma:

\begin{quote}
\textbf{Discursive Dilemma}. A court of three judges has to decide on the liability of a defendant under the charge of breach of contract. An individual is considered liable if there was a valid contract and her behaviour was such as to be considered a breach of the contract. The court takes three majority decisions on the following issues: there was a valid contract ($\alpha$), the individual broke the contract ($\beta$), the defendant is liable ($\alpha \wedge \beta$). Consider the following situation:

\begin{center}
 \begin{tabular}{cccc}
 & $\alpha$ & $\beta$ & $\alpha\wedge\beta$ \\
\toprule
Judge 1 & yes & yes & yes \\
Judge 2 & no & yes & no \\
Judge 3 & yes & no & no \\
\midrule
Majority & yes & yes & no \\
 \end{tabular}
\end{center}

All judges express consistent judgments: they accept the third proposition if and only if the first two are accepted. However, even if there is a majority of judges who believe that there was a valid contract, and even if there is a majority of judges who believe that the individual broke the contract, the individual is considered \emph{not liable} by a majority of the individuals. 
\end{quote}



\subsubsection{Judgment Aggregation}\label{sec:JAtranslation}\index{judgment aggregation}

Judgement aggregation (JA) considers problems in which a finite set of individuals $\N$ has to generate a collective judgment over a set of interconnected propositional formulas \citep[see, e.g.,][]{ListPuppe2009}.
Formally, given a propositional language $\Lang$, an \emph{agenda} is a finite nonempty subset $\Phi\subseteq\Lang$ that does not contain doubly-negated formulas and is closed under complementation (i.e, $\alpha\in\Phi$ whenever $\neg\alpha\in\Phi$, and $\neg\alpha\in\Phi$ for non-negated $\alpha\in\Phi$). 

Each individual in $\N$ expresses a \emph{judgment set}\index{judgment set} $J\subseteq \Phi$, as the set of those formulas in the agenda that she judges to be true. Every individual judgment set $J$ is assumed to be \emph{complete} (i.e., for each $\alpha\in\Phi$ either $\alpha$ or its complement are in $J$) and \emph{consistent} (i.e., there exists an assignment that makes all formulas in $J$ true).
Denote by $\mathcal J (\Phi)$ the set of all complete and consistent subsets of $\Phi$.
Given a finite agenda $\Phi$ and a finite set of individuals~$\N$, a \emph{JA procedure} for $\Phi$ and $\N$ is a function $F: \mathcal J (\Phi)^\N \to 2^{\Phi}$.



\subsubsection{Translation}\label{sec:icphi}

Given a judgment aggregation framework defined by an agenda $\Phi$ and a set of individuals $\N$, let us now construct a setting for binary aggregation with integrity constraints that interprets it.
Let the set of issues $\I_{\Phi}$ be equal to the set of formulas in $\Phi$. The domain $\D_{\Phi}$ of aggregation is therefore $\{0,1\}^{|\Phi|}$. 
In this setting, a binary ballot $B$ corresponds to a judgment set: $B_{\alpha}=1$ if and only if $\alpha\in J$. Given this representation, we can associate with every JA procedure for $\Phi$ and $\N$ a binary aggregation procedure on a subdomain of~$\D_{\Phi}^\N$.

It is important to remark that this is not exactly the standard way of interpreting JA in binary aggregation. The embedding that is given, for instance, by \citet{DokowHolzman2009, DokowHolzman2008}, associates with every judgment set a binary ballot over a set of issues representing only the \emph{positive} formulas in $\Phi$, considering a rejection of the issue associated with a formula $\phi$ as an acceptance of its negation~$\neg\phi$. The same embedding is given by \citet[][Section~2.3]{ListPuppe2009}. 
In our translation we made the choice of introducing both an issue for $\phi$ and one for $\neg\phi$, adding an additional integrity constraint to enforce the completeness of a judgment set. This allows us to easily generalise the framework to the case of incomplete ballots \citep[see, e.g.,][]{DietrichList2008}, without having to resort to an additional symbol for abstention \citep[as is done, e.g., by][]{DokowHolzman2010}

As we did for the case of preference aggregation, we now define a set of integrity constraints for $\D_{\Phi}$ to enforce the properties of consistency and completeness of individual judgment sets. 
Recall that the propositional language is constructed in this case on $|\Phi|$ propositional symbols $p_\alpha$, one for every $\alpha\in\Phi$. 
Call an inconsistent set of formulas each proper subset of which is consistent \emph{minimally inconsistent set} (mi-set). Let $\ic_\Phi$\index{integrity constraint!$c$@$\ic_\Phi$} be the following set of integrity constraints:

\begin{description}  
\item \textbf{Completeness}:  $p_\alpha{\vee} p_{\neg\alpha}$ for all $\alpha\in\Phi$
\item  \textbf{Consistency}: $\neg(\bigwedge_{\alpha\in S} p_\alpha)$ for every mi-set $S\subseteq\Phi$
\end{description}

\noindent
While the interpretation of the first formula is straightforward, we provide some further explanation for the second one. If a judgment set $J$ is inconsistent, then it contains a minimally inconsistent set, obtained by sequentially deleting one formula at the time from $J$ until it becomes consistent. This implies that the constraint previously introduced is falsified by the binary ballot that represents $J$, as all issues associated with formulas in a mi-set are accepted. \emph{Vice versa}, if all formulas in a mi-set are accepted by a given binary ballot, then clearly the judgment set associated with it is inconsistent.


In conclusion, the same kind of correspondence we have shown for SWFs holds between complete and consistent JA procedures and binary aggregation procedures that are collectively rational with respect to $\ic_\Phi$.


\subsubsection{The Discursive Dilemma in Binary Aggregation}

The same procedure that we have used to show that the Condorcet paradox is an instance of our general definition of paradox applies here for the case of the discursive dilemma. 
Let $\Phi$ be the agenda $\{\alpha, \beta, \alpha\wedge\beta \}$, in which we have omitted negated formulas, as for any $J\in\mathcal{J}(\Phi)$ their acceptance can be inferred from the acceptance of their positive counterparts. Consider the profile $\B$ for $\I_\Phi$ described in Table~\ref{table:discu}.
Every individual ballot satisfies $\ic_\Phi$, while the outcome obtained by using the majority rule contradicts one of the constraints of consistency, namely $\neg(p_{\alpha}\wedge p_{\beta}\wedge p_{\neg (\alpha\wedge\beta)})$. Hence, $(\text{\emph{Maj}},\B,\ic_{\Phi})$ constitutes a paradox by Definition~\ref{def:paradox}.

\begin{table}
\begin{center}
\label{table:discu}
\begin{tabular}{cccc}
\hline\noalign{\smallskip}
 & $\alpha$ & $\beta$ & $\alpha\wedge\beta$ \\
\noalign{\smallskip}\hline\noalign{\smallskip}
Judge 1 & 1 & 1 & 1 \\
Judge 2 & 0 & 1 & 0 \\
Judge 3 & 1 & 0 & 0 \\
\noalign{\smallskip}\hline
\emph{Maj} & 1 & 1 & 0 
 \end{tabular}
\end{center}
\caption{The discursive dilemma in binary aggregation.}
\end{table}


\subsection{The Ostrogorski Paradox}\label{sec:ostrogorski}\index{paradox!Ostrogorski paradox}


Another paradox listed by \citet{Nurmi} as one of the main paradoxes of the majority rule on multiple issues is the Ostrogorski paradox. \citet{Ostrogorski} published a treaty in support of procedures inspired by direct democracy, pointing out several fallacies that a representative system based on party structures can encounter. \citet{RaeDaudt1976} later focused on one such situation, presenting it as a paradox or a dilemma between two equivalently desirable procedures (the direct and the representative one), giving it the name of ``Ostrogorski paradox". This paradox, in its simplest form, occurs when a majority of individuals are supporting a party that does not represent the view of a majority of individuals on a majority of issues.

\begin{quote}
\textbf{Ostrogorski Paradox}.
Consider the following situation: there is a two party contest between the Mountain Party (MP) and the Plain Party (PP); three individuals (or, equivalently, three equally big groups in an electorate) will vote for one of the two parties if their view agrees with that party on a majority of the three following issues: economic policy ($E$), social policy ($S$), and foreign affairs policy ($F$).
Consider the following situation:

\begin{center}

 \begin{tabular}{ccccc}
 & $E$ & $S$ & $F$ & Party supported\\
\toprule
Voter 1 & MP & PP & PP & PP \\
Voter 2 & PP & PP & MP & PP\\
Voter 3 & MP & PP & MP & MP\\ 
\midrule
\emph{Maj} & MP & PP & MP & PP 
\end{tabular}

\label{table:ostrogorski}
\end{center}

The result of the two party contest, assuming that the party that has the support of a majority of the voters wins, declares the Plain Party the winner. However, a majority of individuals support the Mountain Party both on the economic policy E \emph{and} on the foreign policy F. Thus, the elected party (the PP) is in disagreement with a majority of the individuals on a majority of the issues. 
\end{quote}

\noindent
\citet{Bezembinder1985} generalised this paradox, defining two different rules for compound majority decisions. The first, the representative outcome, outputs as a winner the party that receives support by a majority of the individuals. The second, the direct outcome, outputs the party that receives support on a majority of issues by a majority of the individuals. An instance of the Ostrogorski paradox occurs whenever the outcome of these two procedures differ.

Stronger versions of the paradox can be devised, in which the losing party represents the view of a majority on \emph{all} the issues involved (see, e.g., Rae and Daudt, \citeyear{RaeDaudt1976}; see also our Table~\ref{table:ostrogorski2}). Further studies of the ``Ostrogorski phenomenon" have been carried out by \citet{DebKelsey1987} as well as by \citet{EckertKlamer2009}.
The relation between the Ostrogorski paradox and the Condorcet paradox has been investigated in several papers \citep{Kelly1989, RaeDaudt1976}, while a comparison with the discursive dilemma was carried out by \citet{Pigozzi2005}.



\subsubsection{The Ostrogorski Paradox in Binary Aggregation}

In this section, we provide a binary aggregation setting that represents the Ostrogorski paradox as a failure of collective rationality with respect to a suitable integrity constraint. 

Let $\{E,S,F\}$ be the set of issues at stake, and let the set of issues $\I_O= \{E,S,F,A\}$ consist of the same issues plus an extra issue $A$ to encode the support for the first party (MP).
A binary ballot over these issues represents the individual view on the three issues $E$, $S$ and $F$: if, for instance, $b_E=1$, then the individual supports the first party MP on the first issue $E$. Moreover, it also represents the overall support for party MP (in case issue $A$ is accepted) or PP (in case $A$ is rejected).
In the Ostrogorski paradox, an individual votes for a party if and only if she agrees with that party on a majority of the issues. This rule can be represented as a rationality assumption by means of the following integrity constraint~$\ic_O$:

$$p_A\leftrightarrow [(p_E\wedge p_S) \vee (p_E\wedge p_F) \vee (p_S\wedge p_F)]$$ 

\noindent
An instance of the Ostrogorski paradox can therefore be represented by the profile $\B$ described in Table~\ref{table:ostro}. 
Each individual accepts issue $A$ if and only if she accepts a majority of the other issues. However, the outcome of the majority rule is a rejection of issue $A$, even if a majority of the issues gets accepted by the same rule. Therefore, the triple $(\text{\emph{Maj}}, \B, \ic_O)$ constitutes a paradox by Definition~\ref{def:paradox}.

\begin{table}[]
\begin{center}
\label{table:ostro}
\begin{tabular}{ccccc}
\hline\noalign{\smallskip}
 & $E$ & $S$ & $F$ & $A$\\
\noalign{\smallskip}\hline\noalign{\smallskip}
Voter 1 & 1 & 0 & 0 & 0 \\
Voter 2 & 0 & 0 & 1 & 0\\
Voter 3 & 1 & 0 & 1 & 1\\
\noalign{\smallskip}\hline
\emph{Maj} & 1 & 0 & 1 & 0
 \end{tabular}
\end{center}
\caption{The Ostrogorski paradox in binary aggregation.}
\end{table}

Using this formalism we can easily devise a strict version of the Ostrogorski paradox, in which the winning party disagrees on a majority of the issues with \emph{all} the individuals. Such a profile is described in Table~\ref{table:ostrogorski2}.

\begin{table}[]
\begin{center}
\label{table:ostrogorski2}
\begin{tabular}{ccccc}
\hline\noalign{\smallskip}
 & $E$ & $S$ & $F$ & $A$\\
\noalign{\smallskip}\hline\noalign{\smallskip}
Voter 1 & 1 & 0 & 0 & 0 \\
Voter 2 & 0 & 1 & 0 & 0\\
Voter 3 & 0 & 0 & 1 & 0\\
Voter 4 & 1 & 1 & 1 & 1\\
Voter 5 & 1 & 1 & 1 & 1\\ 
\noalign{\smallskip}\hline
\emph{Maj} & 1 & 1 & 1 & 0\\
 \end{tabular}
\end{center}
\caption{Strict version of the Ostrogorski paradox in binary aggregation.}

\end{table}



\subsection{The Common Structure of Paradoxical Integrity Constraints}\label{sec:commonstructure}


We can now make a crucial observation concerning the syntactic structure of the integrity constraints that formalise the paradoxes we have presented so far. 
First, for the case of the Condorcet paradox, we observe that the formula encoding the transitivity of a preference relation is the implication $p_{ab}\wedge p_{bc} \rightarrow p_{ac}$. 
This formula is equivalent to $\neg p_{ab} \vee \neg p_{bc} \vee p_{ac}$, which is a clause of size~3, i.e., it is a disjunction of three different literals. 
Second, the formula which appears in the translation of the discursive dilemma is also equivalent to a clause of size~3, namely $\neg p_{\alpha} \vee  \neg p_{\beta} \vee  \neg p_{\neg (\alpha\wedge\beta)} $. 
Third, the formula which formalises the majoritarian constraint underlying the Ostrogorski paradox, is equivalent to the following conjunction of clauses of size~3:
\begin{eqnarray*}
(p_A\vee\neg p_E\vee \neg p_F)\wedge (p_A\vee\neg p_E\vee \neg p_S)\wedge (p_A\vee\neg p_S\vee \neg p_F)\wedge \\ 
\wedge (\neg p_A\vee p_E\vee  p_F) \wedge(\neg p_A\vee p_E\vee  p_S)\wedge (\neg p_A\vee p_S\vee  p_F)
\end{eqnarray*}


\noindent 
Thus, we observe that the \textbf{integrity constraints formalising the most classical paradoxes in aggregation theory all feature a clause (i.e., a disjunction) of size at least~3.}\footnote{This observation is strongly related to a result by \citet{NehringPuppe2007} in the framework of judgment aggregation, which characterises the set of paradoxical agendas for the majority rule as those agendas containing a minimal inconsistent subset of size at least~3.}


\subsection{Further Paradoxes on Multiple Issues}\label{sec:mep}


In this section we describe two further paradoxes that can be analysed using our framework of binary aggregation with integrity constraints: the paradox of divided government and the paradox of multiple elections. Both situations concern a paradoxical outcome obtained by using the majority rule on an aggregation problem defined on multiple issues. The first paradox can be seen as an instance of a more general behaviour described by the second paradox.


\subsection{The Paradox of Divided Government}

The paradox of divided government is a failure of collective rationality that was pointed out for the first time by \citet{BramsEtAl1993}. Here we follow the presentation of \citet{Nurmi1997}. 

\begin{quote}\textbf{The paradox of divided government}.
Suppose that 13 voters (equivalently, groups of voters) can choose for Democratic (D) or Republican (R) candidate for the following three offices: House of Representatives (H), Senate (S) and the governor (G). It is a common assumption that in case the House of Representatives gets a Republican candidate, then at least one of the remaining offices should go to Republicans as well. Consider now the following profile:

\begin{center}
\begin{tabular}{lcccc}
 & $H$ & $S$ & $G$ \\
\toprule
Voters 1-3 & D & D & D \\
Voter 4 & D & D & R \\
Voter 5 & D & R & D \\
Voter 6 & D & R & R \\
Voters 7-9 & R & D & R \\ 
Voters 10-12 & R & R & D\\ 
Voter 13 & R & R & R \\ 
\midrule
\emph{Maj} & R & D & D\\
\end{tabular}
\end{center}

In this situation it is exactly the combination that had to be avoided (i.e., RDD) that is elected, even if no individual voted for it. 
\end{quote}

\noindent
This paradox can be easily seen as a failure of collective rationality: it is sufficient to replace the letters D and R with 0 and 1, and to formulate the integrity constraint as $\neg(p_H \wedge \neg p_S \wedge \neg p_G$). The binary ballot $(1,0,0)$ is therefore ruled out as irrational, encoding the combination (R,D,D) that needs to be avoided.

This type of paradox can be observed in cases like the elections of a committee. Even if it is recognised by every individual that a certain committee structure is unfeasible (i.e., it will not work well together), this may be the outcome of aggregation if the majority rule is being used. 



\subsection{The Paradox of Multiple Elections}\index{paradox!of multiple elections}

Whilst the Ostrogorski paradox was devised to stage an attack against representative systems of collective choice based on parties, the paradox of multiple elections (MEP) is based on the observation that when voting directly on multiple issues, a combination that was not supported nor liked by any of the voters can be the winner of the election \citep{BramsEtAl1998, LacyNiou2000}. While the original model takes into account the full preferences of individuals over combinations of issues, if we focus on only those ballots that are submitted by the individuals, then an instance of the MEP can be represented as a paradox of collective rationality. Let us consider a simple example.

\begin{quote}
\textbf{Multiple election paradox}. 
Suppose three voters need to take a decision over three binary issues $A$, $B$ and $C$. Their ballots are described in the following table:

\begin{center}
\begin{tabular}{ccccc}

 & $A$ & $B$ & $C$ \\
\toprule
Voter 1 & 1 & 0 & 1 \\
Voter 2 & 0 & 1 & 1 \\
Voter 3 & 1 & 1 & 0 \\
\midrule
\emph{Maj} & 1 & 1 & 1\\

\end{tabular}	
\end{center}

The outcome of the majority rule in this situation is the acceptance of all three issues, even if this combination was not voted for by any of the individuals. 
\end{quote}

\noindent
While there seems to be no integrity constraint directly causing this paradox, we may represent the profile in the example above as a situation in which the three individual ballots are bound by a budget constraint $\neg(p_A \wedge p_B \wedge p_C)$. Even if all individuals are giving acceptance to two issues each, the result of the aggregation is the unfeasible acceptance of all three issues. 

As can be deduced from our previous discussion, every instance of the MEP gives rise to several instances of a binary aggregation paradox for Definition~\ref{def:paradox}. 
To see this, it is sufficient to find an integrity constraint that is satisfied by all individuals and not by the outcome of the aggregation.\footnote{Such a formula always exists. Consider for instance the disjunction of the formulas specifying each of the individual ballots. This integrity constraint forces the result of the aggregation to be equal to one of the individual ballots on the given profile, thus generating a binary aggregation paradox from a MEP.} On the other hand, every instance of Definition~\ref{def:paradox} in binary aggregation represents an instance of the MEP, as the irrational outcome cannot have been voted for by any of the individuals. 


In their paper, \citet{BramsEtAl1998} provide many versions of the multiple election paradox, varying the number of issues and the presence of ties. 
\citet{LacyNiou2000} enrich the model by assuming that individuals have a preference order over combinations of issues and submit just their top candidate for the election. They present situations in which, e.g., the winning combination is a Condorcet loser (i.e., it loses in pairwise comparison with all other combinations). 
Some answers to the problem raised by the MEP have already been proposed in the literature on Artificial Intelligence. For instance, a number of papers have studied the problem of devising sequential elections to avoid the MEP in case the preferences of the individuals over combinations of multiple issues are expressed in a suitable preference representation language \citep{XiaEtAl2011, ConitzerXia2012}.


\section{The Majority Rule: Characterisation of Paradoxes}\label{sec:majority}


In this section we generalise the observation made in the Section~\ref{sec:commonstructure} to a full theorem, characterising the class of integrity constraints that are lifted by the majority rule as those formulas that can be expressed as a conjunction of clauses (i.e., disjunctions) of maximal size~2.\footnote{The theorem presented in this section was published in our previous work \citep{GrandiEndrissAIJ2013} as part of a more systematic analysis of collective rationality and axiomatic conditions on aggregation procedures. Here we present a simpler version of the result and we provide a direct proof.}
Our characterisation may be considered a ``syntactic counterpart'' of a result by \citet{NehringPuppe2007} in judgment aggregation, characterising agendas on which the majority rule outputs a consistent outcome as those agendas that only contain minimally inconsistent subsets of maximal size~2.


Let us first provide a formal definition of the majority rule. 
Let $N^\B_j$ be the set of individuals that accept issue $j$ in profile $\B$.
In case the number of individuals is odd, the \emph{majority rule} (\emph{Maj}) has a unique definition by accepting issue $j$ if and only if $|N^\B_j|\geq\frac{n+1}{2}$. For the remainder of this section we make the assumption that the number of individuals is~odd. Recall that an aggregation procedure is collectively rational (CR) with respect to an integrity constraint $\ic$ if does not generate any paradox with $\ic$ (cf. Definition~\ref{def:CR}).


\begin{theorem}\label{thm:icmajority}
The majority rule Maj is CR with respect to $\ic$ if and only if $\ic$ is equivalent to a conjunction of clauses of maximal size~2.
\end{theorem}

\begin{proof}
$(\leftarrow)$ Let $\ic$ be equivalent to a conjunction of clauses of maximal size~2, which we indicate as $\psi=\bigwedge_k D_k$.
We want to show that \emph{Maj} is CR wrt. $\ic$.
We first make the following two observations.
First, since two equivalent formulas define the same set of rational ballots, showing that \emph{Maj} is CR wrt. $\ic$ is equivalent to showing that \emph{Maj} is CR wrt. $\psi$. 
Second, if the majority rule is collectively rational wrt. two formulas $\phi_1$ and $\phi_2$ then it is also CR wrt. their conjunction $\phi_1\wedge \phi_2$. Thus, it is sufficient to show that \emph{Maj} is CR wrt. all clauses $D_j$ to conclude that \emph{Maj} is CR wrt. their conjunction and hence with $\ic$. 
Recall that all clauses $D_j$ have maximal size~2. 
The case of a clause of size~1 is easily solved. Suppose $D_k=p_{j_k}$ or $D_k=\neg p_{j_k}$. Since all individuals must be rational the profile will be unanimous on issue $j_k$, and thus the majority will behave accordingly on issue $j_k$, in accordance with the constraint $D_k$.
Let us then focus on a clause $\ic=\ell_j\vee \ell_k$, where $\ell_j$ and $\ell_k$ are two distinct literals, i.e., atoms or negated atoms. 
A paradoxical profile for the majority rule with respect to this integrity constraint features a first majority of individuals not satisfying literal $\ell_j$, and a second majority of individuals not satisfying literal $\ell_k$. 
By the pigeonhole principle these two majorities must have a non-empty intersection, i.e., there exists one individual that does not satisfy both literals $\ell_j$ and $\ell_k$, but this is incompatible with the requirement that all individual ballots satisfy $\ic$. 


$(\Rightarrow)$
Let us now assume for the sake of contradiction that $\ic$ is not equivalent to a conjunction of clauses of maximal size~2. We will now build a paradoxical situation for the majority rule with respect to $\ic$.

We need the following crucial definition:
Call \emph{minimally falsifying partial assignment} (mifap-assignment) for an integrity constraint $\ic$ an assignment to some of the propositional variables that cannot be extended to a satisfying assignment, although each of its proper subsets can.
We now associate with each mifap-assignment $\rho$ for $\ic$ a conjunction $C_\rho= \ell_1\wedge\dots\wedge \ell_k$, where $\ell_i = p_i$ if $\rho(p_i)= 1$ and $\ell_i=\neg p_i$ if $\rho(p_i)=0$ for all propositional symbols $p_i$ on which $\rho$ is defined. The conjunction $C_\rho$ represents the mifap-assignment $\rho$ and it is clearly inconsistent with $\ic$.
The negation of $C_\rho$ is hence a disjunction, with the property of being a minimal clause implied by $\ic$. Such formulas are known in the literature on knowledge representation as the \emph{prime implicates} of $\ic$, and it is a known result that every propositional formula is equivalent to the conjunction of its prime implicates \citep[see, e.g.,][]{Marquis2000}. 
Thus, we can represent $\ic$ with the equivalent formula $\bigwedge _\rho \neg C_\rho$ of all mifap-assignments $\rho$ for $\ic$. From our initial assumption we can infer that at least one mifap-assignment $\rho^*$ has size $>2$, for otherwise $\ic$ would be equivalent to a conjunction of 2-clauses.

We are now ready to show a paradoxical situation for the majority rule with respect to $\ic$.
Consider the following profile. Let $y_1, y_2, y_3$ be three propositional variables that are fixed by $\rho^*$. 
Let the first individual $i_1$ accept the issue associated with $y_1$ if $\rho(y_1)=0$, and reject it otherwise, i.e., let $b_{1,1}=1-\rho^*(y_1)$. 
Furthermore, let $i_1$ agree with $\rho^*$ on the remaining propositional variables. By minimality of $\rho^*$, this partial assignment can be extended to a satisfying assignment for $\ic$, and let $B_{i_1}$ be such an assignment.
Repeat the same construction for individual $i_2$, this time changing the value of $\rho^*$ on $y_2$ and extending it to a satisfying assignment to obtain $B_{i_2}$. 
The same construction for $i_3$, changing the value of $\rho^*$ on issue $y_3$ and extending it to a satisfying assignment $B_{i_3}$. 
Recall that there are at least $3$ individuals in $\N$. If there are other individuals, let individuals $i_{3s+1}$ have the same ballot $B_{i_1}$, individuals $i_{3s+2}$ ballot $B_{i_2}$ and individuals $i_{3s+3}$ ballot $B_{i_3}$.
The basic profile for $3$ issues and $3$ individuals is shown in Table~\ref{table:majority}.
In this profile, which can easily be generalised to the case of more than 3 individuals, there is a majority supporting $\rho^*$ on every variable on which $\rho^*$ is defined. Since $\rho^*$ is a mifap-assignment and therefore cannot be extended to an assignment satisfying $\ic$, the majority rule in this profile is not collectively rational with respect to $\ic$.\qed

\begin{table}[]
\begin{center}
\label{table:majority}
\begin{tabular}{llll}
\hline\noalign{\smallskip}
& $y_1$ & $y_2$ & $y_3$ \\
\noalign{\smallskip}\hline\noalign{\smallskip}
$i_1$ & 1-$\rho^*(y_1)$ & $\rho^*(y_2)$ & $\rho^*(y_3)$ \\
	$i_2$ & $\rho^*(y_1)$ & 1-$\rho^*(y_2)$ & $\rho^*(y_3)$ \\
	$i_3$ & $\rho^*(y_1)$ & $\rho^*(y_2)$ & 1-$\rho^*(y_3)$ \\
\noalign{\smallskip}\hline
	\emph{Maj} & $\rho^*(y_1)$ & $\rho^*(y_2)$ & $\rho^*(y_3)$
 \end{tabular}

\end{center}
\caption{A general paradox for the majority rule wrt. a clause of size~3.}
\end{table}

\end{proof}

\noindent
In case the number of individuals is even the majority rule does not have a unique definition, to account for the case in which exactly half of the voters accept an issue and half of the voters reject it, and a characterisation along the lines of Theorem~\ref{thm:icmajority} cannot be proven. We refer to our previous work \citep{GrandiEndrissAIJ2013} for a more detailed analysis of the set of integrity constraints that are lifted by the majority rule for an even number of individuals.



\section{Conclusions}\label{sec:conclusions}


The first conclusion that can be drawn from this paper dedicated to paradoxes of aggregation is that the majority rule is to be avoided when dealing with collective choices over multiple issues. 
This fact stands out as a counterpart to May's Theorem (\citeyear{May52}), which proves that the majority rule is the only aggregation rule \emph{for a single binary issue} that satisfies a set of highly desirable conditions. The sequence of paradoxes we have analysed in this paper shows that this is not the case when multiple issues are involved. 
While this fact may not add anything substantially new to the existing literature, the wide variety of paradoxical situations encountered in this paper stresses even further the negative features of the majority rule on multi-issue domains.
%

A second more significant conclusion is that a large number of paradoxes of Social Choice Theory share a common structure, and that this structure is formalised by our Definition~\ref{def:paradox}, which stands out as a truly general definition of paradox in aggregation theory. 
Moreover, by analysing the integrity constraints that underlie some of the most classical paradoxes, we were able to identify a common syntactic feature of paradoxical constraints. 
Starting from this observation, we have provided a full characterisation of the integrity constraints that are lifted by the majority rule, as those formulas that are equivalent to a conjunction of clauses of size at most~2. 


To the best of our knowledge Definition~\ref{def:paradox} is the first general definition of paradox in aggregation theory. 
However, the paradoxical situations presented in this paper constitute a fragment of the problems that can be encountered in the formalisation of collective choice problems. For instance, paradoxical situations concerning voting procedures \citep{Nurmi, Saari1989}, which take as input a set of preferences and output a set of winning candidates, are not included in our analysis.
Recent work on paradoxes of aggregation pointed at similarities within different frameworks, e.g., comparing two such examples \citep{Pigozzi2005}, or proposing a geometric approach for the study of paradoxical situations \citep{EckertKlamer2009} using the theoretical setting developed by \citet{Saari}. 
There also exists a considerable amount of work exploring the relation between preference aggregation and other frameworks of aggregation \citep{ListPettit2004, DietrichList2007a, Grossi2009, Grossi2010, Porello2010}. 
Our approach is similar to that of \citet{DietrichList2007a}. In their work, the authors embed the framework of preference aggregation into the framework of judgment aggregation in general logics \citep{Dietrich2007} by using a simple first-order logic of orders. While their aim is more theoretical and their setting is more general, our analysis has the advantage of being based on a very simple logical setting, without losing the generality which is needed for the analysis of such a wide range of paradoxical situations. By doing so we aim at stressing the propositional (i.e., binary) nature of many paradoxes which were encountered in the study of social choice.

The last conclusive statement we would like to put forward regards the interpretation of some of the paradoxes presented in this chapter. 
We have already remarked how some of these examples have been employed in the literature to show weaknesses and advantages of either the direct approach to democratic choice (represented by issue-by-issue aggregation) or the representative one. 
In particular the two paradoxes presented in Section~\ref{sec:mep} (the paradox of divided government and the MEP) seem to suggest that direct decisions over multiple issues should be avoided, at least when issues are not completely independent from one another. 
In our view, elections over multi-issue domains cannot be escaped: not only do they represent a model for the aggregation of more complex objects like preferences and judgments, as seen in Section~\ref{sec:paba} and~\ref{sec:jaba}, but they also stand out as one of the biggest challenges to the design of more complex automated systems for collective decision making. 
%
A crucial problem in the modelling of real-world situations of collective choice is that of identifying the set of issues that best represent a given domain of aggregation, and devising an integrity constraint that models correctly the correlations between those issues.
This problem obviously represents a serious obstacle to a mechanism designer, and is moreover open to manipulation.
However, a promising direction for future work consists in structuring collective decision problems with more detailed models \emph{before} the aggregation takes place, e.g., by discovering a shared order of preferential dependencies between issues \citep{LangXia2009, AiriauEtAlIJCAI2011}, facilitating the definition of collective choice procedures on complex domains without having to elicit the full preferences of individuals.
Such models can be employed in the design and the implementation of automated decision systems, in which a safe aggregation, i.e., one that avoids paradoxical situations, is of the utmost necessity.

\subsection*{Acknowledgements}
Earlier versions of this work have been presented at the Dagstuhl seminar on Computation and Incentives in Social Choice in 2012, and at the Fourth International Workshop on Computational Social Choice (COMSOC-2012) in Krak\'ow. I~would like to thank the audience of both venues for their useful suggestions. I~am indebted to Ulle Endriss and J\'er\^ome Lang for their guidance and for their useful comments on this work.

\bibliography{paradoxes}

\end{document}